%% file: main.tex
\documentclass[12pt, onecolumn,draftclsnofoot]{IEEEtran}%
\usepackage[utf8]{inputenc}
\input{preamble}

\usepackage{tikz-cd} 

\title{Lossless digraph signal processing via polar decomposition}
\author{Feng~Ji %
\thanks{The author is with the School of Electrical and Electronic Engineering, Nanyang Technological University, 639798, Singapore (e-mail: jifeng@ntu.edu.sg).}%
}

\begin{document}
    
\maketitle

\begin{abstract}
In this paper, we present a signal processing framework for directed graphs. Unlike undirected graphs, a graph shift operator such as the adjacency matrix associated with a directed graph usually does not admit an orthogonal eigenbasis. This makes it challenging to define the Fourier transform. Our methodology leverages the polar decomposition to define two distinct eigendecompositions, each associated with different matrices derived from this decomposition. We propose to extend the frequency domain and introduce a Fourier transform that jointly encodes the spectral response of a signal for the two eigenbases from the polar decomposition. This allows us to define convolution following a standard routine. Our approach has two features: it is lossless as the shift operator can be fully recovered from factors of the polar decomposition. Moreover, it subsumes the traditional graph signal processing if the graph is directed. We present numerical results to show how the framework can be applied. 
\end{abstract}


\section{Introduction} \label{sec:int}
The field of graph signal processing (GSP) is experiencing significant growth and focuses on the analysis of signals defined on graphs \cite{Shu13, San13, Tsi15, OrtFroKov:J18, Ji19, LeuMarMou:J23, JiaFenTay:J23, Ji22}. Numerous real-world phenomena, including social networks, transportation systems, and sensor networks, can be naturally represented as graphs. In GSP, the key idea revolves around the vector space of graph signals, where each node in a given graph is assigned a numerical value by a graph signal. As signals on a graph form a vector space, GSP employs linear transformations, such as the graph Fourier transform (GFT) and graph filters, to examine and explore relationships among graph signals.

On the other hand, despite a growing need for applications that involve directed graphs 
(digraphs)~\cite{sarcheshmehpour2021federated, giraldo2020semi}, 
the intrinsic asymmetry of their edge connections presents challenges in defining GFTs. First, the basic graph operators become non-diagonalizable for a large number of digraphs. This makes it challenging to define key concepts such as the frequency domain and convolution. Second, ``ordering" the graph Fourier basis, so that we can interpret them as low or high frequencies is difficult. This is because the eigenvalues, which can be used for the ordering, can be complex numbers. 

As a consequence, a variety of approaches have been proposed to address these challenges. For example, studies directly using the adjacency matrix as the graph operator have suggested utilizing the Jordan decomposition~\cite{Sandryhaila_2013_2, Domingos_2020} or the singular value decomposition~\cite{Chen_2022} of the adjacency matrix for constructing GFT. Symmetrization of the adjacency matrix has also been studied~\cite{symmetrizations}. \cite{Sardellitti} proposes solving an optimization based on graph directed variation to obtain an orthonormal basis for Fourier transform. 
There are studies that define the graph operator as decomposed matrices of the adjacency matrix, producing diagonalizable matrices~\cite{MHASKAR2018611, unitaryshift}. In \cite{Furutani2019GraphSP}, the authors propose to directly modify the graph Laplacian so that the resulting one is Hermitian and hence diagonalizable.
The recent work \cite{9325908} proposes to modify the graph structure so that the corresponding graph operator on the resulting graph becomes diagonalizable. However, these approaches have their respective shortcomings. For example, the Jordan decomposition is not orthogonal. Hence, signal energy is not preserved after the Fourier transformation, and the basis vectors have non-trivial interactions among themselves. The approaches of optimization and topology perturbation are usually not lossless. The original directed graph topology cannot be faithfully recovered from the basis of the Fourier transform. 

In this paper, we propose a lossless digraph signal process scheme. There are two main ingredients of our approach. For a graph shift operator (GSO) $\bS$, which may not be symmetric, we consider the polar decomposition of $\bS$. This means that $\bS$ is the product of a unitary matrix $\bU$ and symmetric matrix $\bP$, both are orthogonally diagonalizable. Then, we introduce the Fourier transform that utilizes the orthonormal basis of $\bU$ and $\bP$ jointly. Unlike the traditional GSP, the frequency domain is $n^2$ dimensional instead of $n$ dimensional, which allows us to package more information and thus the framework is lossless. Moreover, we also demonstrate that our framework subsumes the traditional GSP if $\bS$ is symmetric, and hence we have a strict generalization of GSP.

The rest of the paper is organized as follows. In \cref{sec:pre}, we review the traditional GSP and matrix polar decomposition, which is the fundamental tool of this paper. In \cref{sec:tft}, we introduce the Fourier transform. The main idea is to extend the frequency domain from $\mathbb{R}^n$ to $\mathbb{C}^{n^2}$. Once we have properly introduced the frequency domain, we define the notion of convolution in \cref{sec:con} following a standard routine. In \cref{sec:grp}, we study how matrix perturbation affects the framework. We discuss miscellaneous topics in \cref{sec:mis} such as how to resolve the issue of non-uniqueness of polar decomposition and eigendecomposition. We present numerical results in \cref{sec:nur} and conclude in \cref{sec:conc}.

\section{Preliminaries} \label{sec:pre}
In this section, we review the basics of graph signal processing (GSP) \cite{Shu13} and polar decomposition \cite{Hal15}. In the paper, we use the term ``GSP'' exclusively for the traditional graph signal processing on undirected graphs. 

Let $G=(V,E)$ be a graph of size $n$. If $G$ is undirected, GSP requires a symmetric graph shift operator (GSO) $\bS$. Common choices of $\bS$ include its adjacency matrix $\bA$ and the Laplacian matrix $\bL = \bA - \bD$, where $\bD$ is the degree matrix. By the spectral theorem, the symmetric matrix $\bS$ admits an eigendecomposition $\bS = \bV\bLambda\bV^{\T}$, where $\bLambda$ is the diagonal matrix consisting of eigenvalues of $\bS$ and the columns of $\bV$ are the corresponding eigenvectors. Thus $\bV$ is an orthogonal matrix. Suppose $V=\{v_1,\ldots,v_n\}$ has a fixed ordering. A graph signal $\boldf = (x_i)_{1\leq i\leq n}$ is a vector in $\mathbb{R}^n$, where $x_i$ is the signal assigned to $v_i$. The graph Fourier transform (GFT) of $\bx$ w.r.t.\ $\bS$ is
\begin{align} \label{eq:wbb}
    \widehat{\boldf} = \bV^{\T}\boldf \in \mathbb{R}^n. 
\end{align}
The frequency domain is indexed by the eigenvalues of $\bS$. 

However, if $G$ is directed, then we have the challenge that $\bS = \bA$ or $\bL$ may not be symmetric. We usually do not have an eigendecomposition and thus important signal processing notions such as GFT (\ref{eq:wbb}) cannot be properly defined. We propose to resolve this issue using matrix polar decomposition \cite{Hal15, Kwa23}, which is reviewed next. 

Recall that for any $n\times n$ matrix $\bS$, its polar decomposition is $\bS = \bU\bP$, where $\bU$ is an orthogonal matrix and $\bP$ is symmetric and positive semi-definite. Moreover, if $\bS$ is invertible, then the decomposition is unique, and $\bP$ is positive definite. Explicitly, let $\bS = \bV\bSigma\bW^{\T}$ be the singular value decomposition of $\bS$ \cite{Tao12}, then $\bU = \bV\bW^{\T}$ and $\bP = \bW\bSigma\bW^{\T}$. 

Both $\bU$ and $\bP$ are normal and thus they admit (complex) eigendecompositions, i.e., $\bU= \bV_1\bLambda_1\bV_1^H$ and $\bP= \bV_2\bLambda_2\bV_2^H$, where $\bV_1, \bV_2$ are unitary, $\bLambda_1, \bLambda_2$ are diagonal and $^H$ is the adjoint operator (Hermitian transpose). Geometrically, $\bP$ performs scaling, while $\bU$ performs a rotation/reflection. By convention, we increasingly order the eigenvalues of $\bLambda_2$ according to the absolute value and those of $\bLambda_1$ according to the complex phase. We shall use $\bV_1,\bV_2$ to define the Fourier transform in the next section. 

\section{The Fourier transform} \label{sec:tft}

In this section, we define the Fourier transform. We use the same notations as in \cref{sec:pre} for the polar decomposition. To motivate, consider a (complex) graph signal $\boldf$. We want to study its shifted signal $\bg = \bS(\boldf) = \bU\bP(\boldf)$. Its magnitude $|\bg|$ is the same as $|\bP(\boldf)|$. Therefore, taking ordinary GFT (\ref{eq:wbb}) w.r.t.\ $\bP$ gives us ``smoothness'' information of $\boldf$. However, rotational information is missing if we ignore $\bU$. To incorporate contributions from both $\bP$ and $\bU$, we follow the idea of \cite{Ji22} and propose to expand the spectral domain from $\mathbb{R}^n$ to $\mathbb{R}^{n}\otimes \mathbb{C}^n \subset \mathbb{C}^{n}\otimes \mathbb{C}^n \cong \mathbb{C}^{n^2}$. Following the ordering of the matrix multiplication, the proposed Fourier transform should encode how $\boldf$ respond w.r.t.\ $\bP$ and how eigenvectors of $\bP$ should respond w.r.t.\ $\bU$ (cf.\ \cref{sec:gmd} below). 

Specifically, we denote the columns of $\bV_1$ by $\{\bv_{1,1},\ldots,\bv_{1,n}\}$ and of $\bV_2$ by $\{\bv_{2,1},\ldots,\bv_{2,n}\}$. They are the eigenvectors of $\bV_1$ and $\bV_2$ respectively. For the signal $\boldf$, we may express it as
\begin{align} \label{eq:bsw}
    \boldf = \sum_{1\leq i\leq n} \widehat{\boldf}_i\bv_{2,i}, 
\end{align}
where $\widehat{\boldf}_i = \langle \bv_{2,i}, \boldf\rangle = \bv_{2,i}^H\boldf$ is the GFT of $\boldf$ w.r.t.\ $\bP$ and $\langle \cdot, \cdot \rangle$ is the complex inner product.  

On the other hand, for each $\bv_{2,i}$, we may inspect its ``angle'' w.r.t.\ to each direction of rotation $\bv_{1,j}$ as $p_{i,j} = \langle \bv_{1,j},\bv_{2,i} \rangle$. Therefore, 
\begin{align} \label{eq:bsa}
    \bv_{2,i} = \sum_{1\leq j\leq n} p_{i,j}\bv_{1,j}. 
\end{align}

Combining (\ref{eq:bsw}) and (\ref{eq:bsa}), we have
\begin{align*}
    \boldf = \sum_{1\leq j\leq n}\Big(\sum_{1\leq i\leq n}\widehat{\boldf}_ip_{i,j}\Big)\bv_{1,j}. 
\end{align*}
Therefore, we define the Fourier transform by considering contributions from $\bV_1$ and $\bV_2$ jointly.

\begin{Definition} \label{defn:tft}
    The \emph{Fourier transform} w.r.t.\ $\bS$ is $\mathcal{F}_{\bS}:\mathbb{R}^n \to \mathbb{C}^{n^2}$ is defined as follows: for a signal $\boldf$, then 
    \begin{align} \label{eq:mfb}
    \mathcal{F}_{\bS}(\boldf)_{i,j} = \widehat{\boldf}_ip_{i,j} = \langle \bv_{2,i},\boldf \rangle \langle \bv_{1,j},\bv_{2,i} \rangle.
    \end{align}
    For simplicity, we usually denote $\mathcal{F}_{\bS}(\boldf)$ by $\mathcal{F}(\boldf)$ or $\widetilde{\boldf}$ if $\bS$ is clear from the context. 
    
    The \emph{inverse transform} $\mathcal{I}_{\bS}: \mathbb{C}^{n^2} \to \mathbb{C}^n$ (or simply $\mathcal{I}$) is 
    \begin{align*}
    (a_{i,j})_{1\leq i,j\leq n} \mapsto \sum_{1\leq j\leq n}\sum_{1\leq i\leq n}a_{i,j}\bv_{1,j}.
    \end{align*}
\end{Definition}

To give an explicit matrix formula for the Fourier transform, let $\bg$ be any (column) vector, and use $D(\bg)$ to denote the diagonal matrix with diagonal $\bg$. Moreover, use $\bc$ to denote the constant (column) vector with each entry $1$. Then we have the following observations.

\begin{Lemma} \label{lem:mbb}
\begin{enumerate}[(a)]
\item $\mathcal{F}(\boldf) = \bV_1^H\bV_2D(\bV_2^H\boldf)$ and $\mathcal{I}(\bM) = \bV_1\bM\bc$ for $\boldf \in \mathbb{R}^n, \bM \in \mathbb{C}^{n^2}$. 
    \item Parseval's identity: $\norm{\boldf} = \norm{\mathcal{F}(\boldf)}$, where $\norm{\cdot}$ is the complex Euclidean norm. 
\end{enumerate}
\end{Lemma}

\begin{proof}
(a) follows directly from the definition. For (b), to show $\norm{\boldf}^2 = \norm{\mathcal{F}(\boldf)}^2$, we compute
    \begin{align*}
        &\norm{\mathcal{F}(\boldf)}^2  = \sum_{1\leq j\leq n}\sum_{1\leq i\leq n} |\widehat{\boldf}_ip_{i,j}|^2 \\
       = &  \sum_{1\leq i\leq n}|\widehat{\boldf}_i|^2\sum_{1\leq j\leq n} |p_{i,j}|^2  = \sum_{1\leq i\leq n}|\widehat{\boldf}_i|^2\norm{\bv_{2,i}}^2 \\
       = &  \sum_{1\leq i\leq n}|\widehat{\boldf}_i|^2 = \norm{\boldf}^2.
    \end{align*}
\end{proof}

The proposed framework recovers that of traditional GSP. To see this, suppose the GSO $\bS$ is normal, e.g., symmetric or orthogonal. In this case, both $\bV_1$ and $\bV_2$ can be chosen to be equal. This implies that $p_{i,j} = \delta_{i,j} = 1$ if $i=j$ and $0$ otherwise. Therefore, $\mathcal{F}(\boldf)_{i,j} = \widehat{\boldf}_i$ if $i=j$ and $0$ otherwise. This is the reason why for undirected graphs or directed cycle graphs, we only need $n$ (diagonal) components in GSP. 

\begin{Example}
In this example, we show sample spectral plots (in \figref{fig:sp}) of the same signal on different directed graphs. Let $G_0, \ldots, G_5$ be directed graphs on $50$ nodes. For the extreme cases, $G_0$ is the undirected cycle graph and $G_5$ is the directed cycle graph. Intermediately, for $0< k <5$, $G_k$ is obtained from $G_0$ by making randomly chosen $10k$ edges directed according to the ordering of the nodes. Shift operators are the (directed) Laplacians $\bL_k, 0\leq k\leq 5$. We choose a bandlimited signal $\boldf$ w.r.t.\ $\bL_0$. The spectral plots of the componentwise absolute values of its Fourier transforms w.r.t.\ $\bL_k, 0\leq i\leq 5$ are shown in \figref{fig:sp}. For $\bL_0$, the result is the same as that traditional GSP offers (placed along the diagonal). If we let $k$ increase, we see that the patterns change gradually. We shall see that this supports the result on graph perturbation in \cref{sec:grp} below. If the graph differs more from $G_0$, then the spectrum of $\boldf$ is more spread out, suggesting the usefulness of having a lossless signal processing framework.  
    \begin{figure}
        \centering
        \includegraphics[scale=0.4, trim={4.5cm 1.5cm 1.5cm 0.5cm},clip]{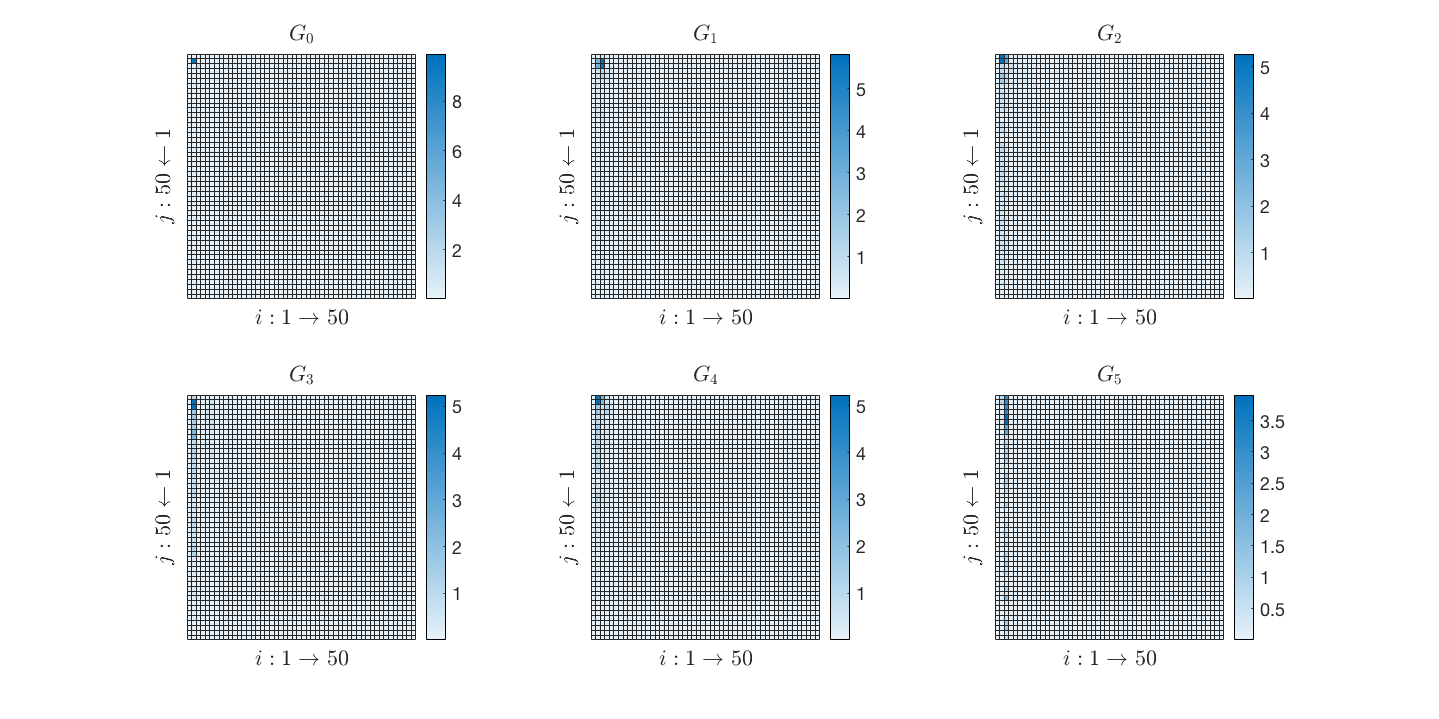}
        \caption{Spectral plots of $\boldf$ w.r.t.\ $\bL_i, 0\leq i\leq 5$.}
        \label{fig:sp}
    \end{figure}
\end{Example}

\section{Convolution} \label{sec:con}
Important tools in GSP are the convolutions \cite{Shu13, Def16, Jif23}. One way to interpret the convolution is to consider it as a polynomial in the GSO $\bS$. We can do so for GSP for the particular reason that $\bS = \bW\bSigma\bW^{\T}$ is diagonalizable and a monomial takes the form $\bS^m = \bW\bSigma^m\bW^{\T}$. Therefore, for the expression $\bSigma^m$, we essentially take entrywise multiplication in the frequency domain. However, this does not hold for an $\bS$ from a directed graph, if it is not diagonalizable. Therefore, we may consider the interpretation that convolution is a multiplication in the frequency domain \cite{Jif23}. For this, we use the identification of the frequency domain with $\mathbb{C}^{n^2}$. More precisely, we define:
\begin{Definition}
    For $\bg,\bh \in \mathbb{C}^{n^2}$, let $\bg \odot \bh \in \mathbb{C}^{n^2}$ be the entrywise multiplication, i.e., the Hadamard product. Moreover, defined $\bh^k = \bh\odot\cdots\odot \bh$ to be the $k$-fold self multiplication. For any $\bh \in \mathbb{C}^{n^2}$, it defines a \emph{convolution} $\bh*: \mathbb{R}^n \to \mathbb{C}^n$ by 
    \begin{align*}
        \bh*\boldf = \mathcal{I}\Big(\widetilde{\boldf} \odot \bh\Big). 
    \end{align*}
\end{Definition}

For the polar decomposition of $\bS$, let the diagonal entries of $\bLambda_1$ (resp.\ $\bLambda_2$) be $(\lambda_{1,i})_{1\leq i\leq n}$ (resp. $(\lambda_{2,i})_{1\leq i\leq n}$).

\begin{Lemma} \label{lem:lbb}
    Let $\bh_{\bS} \in \mathbb{C}^{n^2}$ be defined by $\bh_{\bS,i,j} = \lambda_{2,i}\lambda_{1,j}$. Then $\bh_{\bS}*\boldf = \bS(\boldf)$. 
\end{Lemma}

\begin{proof}
From the definition, we directly evaluate $\bh_{\bS}*\boldf$ as
\begin{align*}
    & \bh_{\bS}*\boldf = \sum_{1\leq j\leq n}\Big(\sum_{1\leq i\leq n}\widehat{\boldf}_ip_{i,j}\lambda_{2,i}\lambda_{1,j}\Big)\bv_{1,j} \\
    =& \sum_{1\leq j\leq n}\lambda_{1,j}\Big(\sum_{1\leq i\leq n}\widehat{\boldf}_ip_{i,j}\lambda_{2,i}\Big)\bv_{1,j} \\
    =& \sum_{1\leq j\leq n}\lambda_{1,j}\langle \bv_{1,j},\sum_{1\leq i\leq n}\widehat{\boldf}_i\lambda_{2,i}\bv_{2,i}\rangle\bv_{1,j} \\
    =& \sum_{1\leq j\leq n}\lambda_{1,j}\langle \bv_{1,j}, \bV_2\bLambda_2\bV_2^H\boldf\rangle\bv_{1,j} \\
    =& \sum_{1\leq j\leq n}\lambda_{1,j}\langle \bv_{1,j},\bP(\boldf)\rangle\bv_{1,j} \\
    =&  \bU\big(\bP(\boldf)\big) = \bS(\boldf). 
\end{align*}
\end{proof}

As $\bU$ is unitary and $\bP$ is positive semi-definite, $\lambda_{2,i}$ is always non-negative and it can be interpreted as the radius, while $\lambda_{1,j}$ is of the form $e^{\bi\theta}$. Hence, $\lambda_{2,i}\lambda_{1,j}$ is the polar form of the complex number $\bh_{i,j}$. As a consequence of \cref{lem:lbb}, if $\bS$ is normal, then $\bh_{\bS}^k*$ is just the monomial convolution filter $\bS^k$. 

We have the freedom to choose $\bh$ to construct different convolution filters. For example, if $\bh$ is a characteristic signal valued in $\{0,1\}$, then we have a bandpass filter. Sampling theory can be developed analogously to \cite{Ji22}. For another construction, consider any function $f: \mathbb{C} \to \mathbb{C}$. We use $f(\bh_{\bS})$ to denote the signal in $\mathbb{C}^{n^2}$ whose $(i,j)$-th entry is $f(\bh_{\bS,i,j}) = h(\lambda_{2,i}\lambda_{1,j})$. Examples of $f$ include polynomials or functions of the form $1/(1+cx)$ (cf.\ \cref{sec:nur}).

\section{Graph perturbation} \label{sec:grp}
In this section, we discuss graph perturbation, which is used in \cite{9325908} to modify a directed graph to obtain a diagonalizable shift operator. Many procedures are considered as a perturbation. For example, adding or removing a small percentage of directed edges is considered as a perturbation. Another example is a small change in the edge weights if the graph is weighted. Though our framework is lossless without modifying the graph structure, we still want to understand how such perturbations affect the Fourier transform and filtering.

In this section, we use $\norm{\cdot}$ to denote the matrix operator norm, and $\norm{\cdot}_F$ to denote the matrix Frobenius norm, the same as the $2$-norm. All norms are equivalent for finite dimensional spaces. For example, for any matrix $\bM$ of rank $r$, we have $\norm{\bM} \leq \norm{\bM}_F \leq \sqrt{r}\norm{\bM}$. 

For another shift operator $\bS'$, denote $\bS'-\bS$ by $\delta \bS'$, as the perturbation matrix. Usually, we are interested in small perturbation, which means $\norm{\delta\bS'}$ is small. In this subsection, we assume that $\bS$ and $\bS'$ are invertible and \emph{non-derogatory}, i.e., both factors in the polar decomposition do not have repeated eigenvalues. If $\norm{\delta\bS'}$ is sufficiently small, then $\bS$ being invertible and non-derogatory implies the same holds for $\bS'$. 

To state the following observation, let $P_m(x,y)$ be the polynomial $(x+y)^m-y^m$. We notice that each of its monomial terms has at least one $x$ factor.

\begin{Lemma} \label{lem:sso}
Suppose sequences of matrices $\bM_1,\ldots,\bM_m$ and $\bM_1',\ldots,\bM_m'$ satisfy: for $1\leq i\leq m$, matrices $\bM_i, \bM_i'$ are of the same size and $\bM = \prod_{1\leq i\leq m} \bM_i$, $\bM' = \prod_{1\leq i\leq m} \bM_i'$ are well-defined. Let $\delta \bM_i' = \bM_i' -\bM_i$. If there are $x,y>0$ such that $\norm{\delta \bM_i'} \leq x, \norm{\bM_i} \leq y, 1\leq i\leq m$, then $\norm{\bM'-\bM} \leq P_m(x,y)$.
\end{Lemma}

\begin{proof}
    Notice that $\norm{\bM' - \bM} = \norm{\prod_{1\leq i\leq m}(\bM_i+\delta \bM_i') - \prod_{1\leq i\leq m}\bM_i}$. We expand the product $\prod_{1\leq i\leq m}(\bM_i+\delta \bM_i')$ and apply the triangle inequality. The norm $\norm{\bM' - \bM}$ is upper bounded by a sum of the form $\norm{\prod_{1\leq i\leq m}\bN_i}$ with $\bN_i$ either $\bM_i$ or $\delta \bM_i'$, and at least one of them is $\delta \bM_i'$. Then as $\norm{\cdot}$ is submultiplicative, we obtain the desired an upper bound $P_m(x)$ by replacing $\bM_i$ with $y$ and $\delta \bM_i'$ by $x$.
\end{proof}

The Fourier transform \cref{defn:tft} and convolutions are linear transformations, and hence it makes sense to talk about their operator norms and Frobenius norms.  

\begin{Theorem} \label{thm:tac}
There are constants $\epsilon$ and $C$ independent of $\delta \bS'$ such that if $\norm{\delta \bS'}_F \leq \epsilon$, then the following holds
\begin{enumerate}[(a)]
    \item $\norm{\mathcal{F}_{\bS}-\mathcal{F}_{\bS'}} \leq P_3(C\norm{\delta\bS'}_F,1)$.
    \item $
        \norm{\bh_{\bS}^k*-\bh_{\bS'}^k*} \leq P_{2k+4}(C\norm{\delta\bS'}_F,\max(1,\lambda_{2,n}))$.
\end{enumerate}
\end{Theorem}

\begin{proof}
Let $\bV_i', \bLambda_i',i=1,2$ be the unitary and diagonal factors resulting from the polar decomposition of $\bS'$. Denote $\delta\bV_i' = \bV_i'-\bV_i, \delta\bLambda_i' = \bLambda_i'-\bLambda_i, i=1,2$. We first claim that there is a constant $C$ independent of $\bS'$ and $n$ such that for $i=1,2$, then the following holds
\begin{enumerate}[(i)]
    \item $\norm{\delta\bV_i'}_F \leq C\norm{\delta\bS'}_F$, and
    \item $\norm{\delta\bLambda_i'}_F \leq C\norm{\delta\bS'}_F$.
\end{enumerate}
The inequalities follow from the matrix perturbation theory as follows. Let $\bS' = \bU'\bP'$ be the unique polar decomposition of the invertible matrix $\bS'$ and write $\delta \bU' = \bU'-\bU, \delta \bP' = \bP'-\bP$. Then by \cite[Theorem~1]{Li95} and \cite{Che89}, we have $\norm{\delta \bU'}_F \leq c_0\norm{\delta\bS'}_F$ and $\norm{\delta \bP'}_F \leq \sqrt{2}\norm{\delta\bS'}_F$, where $c_0$ depends only on the singular values of $\bS$. 

By Mirsky's theorem (\cite[Theorem~2]{Ste90}), we have $\norm{\delta \bLambda_2'}_F \leq \norm{\delta\bP'}_F$. On the other hand, by the Hoffman-Wielandt theorem (\cite[VI.34]{Bha96}) for normal matrices, there is a permutation of $\sigma$ of indices $\{1,\ldots,n\}$ such that 
\begin{align*}
    \sum_{1\leq i\leq n}|\lambda_{1,i}-\lambda_{1,\sigma(i)}'|^2 \leq \norm{\delta \bU'}_F^2 \leq c_0^2\norm{\delta\bS'}_F^2 \leq c_0^2\epsilon^2.
\end{align*}
As all $\lambda_{1,i},1\leq i\leq n$ are distinct by the non-derogatory condition, if $\epsilon$ is sufficiently small, then the closet eigenvalue to $\lambda_{1,i}$ is $\lambda'_{1,i}$. Hence, $\sigma$ is the identity permutation and we have $\norm{\delta \bLambda_1'}_F \leq \norm{\delta\bU'}_F \leq c_0\epsilon$. Therefore, together with the assumption that both $\bU$ and $\bP$ are non-deregatory, we may choose $\epsilon$ small enough such that there is a $\rho>0$ (depending only on $\epsilon$) such that $|\lambda_{1,i}-\lambda_{1,j}'| \geq \rho, |\lambda_{2,i}-\lambda_{2,j}'| \geq \rho$ for any $1\leq i\neq j \leq n$. 

By the Davis-Kahan theorem (\cite[VII.3]{Bha96}) and its generalization (for normal matrices) in \cite[Corollary~5]{Bha22}, there is a constant $c_1$ depending only on $\rho$ such that $\norm{\delta \bV_1'}_F \leq c_1\norm{\delta\bU'}_F$, $\norm{\delta \bV_2'}_F \leq c_1\norm{\delta\bP'}_F$. Setting $C = \max(c_0,\sqrt{2})c_1$, we obtain the claimed inequalities. From the proof, we see that $C$ depends only on the singular values of $\bS$ and $\epsilon$, which depends only on the eigenvalues of $\bU$ and $\bP$. This proves the claim. 

To prove (a) and (b), we apply \cref{lem:sso} using the claim above. For (a), recall that by \cref{lem:mbb}, we have $\mathcal{F}_{\bS}(\boldf) = \bV_1^H\bV_2D(\bV_2^H\boldf)$. Notice that $D(\bV_2^H\boldf)$ is a linear operator with a $n^2$-dimensional codomain. Its norm is $1$ and satisfies the same perturbation bound as $\bV_2^H$. Therefore (a) follows immediately from \cref{lem:sso}. For (b), we observe that $\bh_{\bS}^k*$ can be decomposed into $2k+4$ factors: 
\begin{itemize}
    \item $D(\bV_2^H\boldf)$, $\bV_2$, $\bV_1^H$ from $\mathcal{F}_{\bS}$;
    \item $\bV_1$ from $\mathcal{I}_{\bS}$;
    \item $k$-copies of the transformation $\mathbb{C}^{n^2} \to \mathbb{C}^{n^2}$ that multiplies the $j$-th column by $\lambda_{1,j}, 1\leq j\leq n$.
    \item $k$-copies of the transformation $\mathbb{C}^{n^2} \to \mathbb{C}^{n^2}$ that multiplies the $i$-th row by $\lambda_{2,i}, 1\leq i\leq n$.
\end{itemize}
Each of these factors is bounded in norm by $\max(1,\lambda_{2,n})$ and its perturbation bound is $C\norm{\delta\bS'}_F$ by the above claims. Therefore, (b) of the theorem follows from \cref{lem:sso}. 
\end{proof}

Intuitively, the result claims that if we approximate $\bS$ by a normal $\bS'$ such that traditional GSP tools can be applied, then the framework developed in this paper ``almost'' agrees with the traditional GSP. 

\section{Miscellaneous} \label{sec:mis}
This section contains a few miscellaneous topics. 

\subsection{Different choices of eigenbasis}
The polar decomposition $\bS = \bU\bP$ may not be unique. Moreover, even for fixed $\bU$ or $\bP$, they may have repeated eigenvalues, thus violating the non-derogatory condition of \cref{thm:tac}. As a consequence, there is an ambiguity in the matrices $\bV_1, \bV_2$ needed to define the Fourier transform \cref{defn:tft}. To (partially) resolve this issue, we propose an optimization framework. 

Let $d(\cdot,\cdot)$ be a metric on the space of $n\times n$ unitary matrices $\calU_n$ such as the metric induced by the operator norm. For a polar decomposition $\bS = \bU\bP$, we propose to construct the unitary matrices of eigenvectors $\bV_1,\bV_2$ by solving the following optimization
\begin{align} \label{eq:mbv}
    \min_{\bV_1,\bV_2 \in \calU_n} d(\bV_1,\bV_2), \text{ provided } \bV_1,\bV_2
 \text{ form eigenbasis of } \bU,\bP. \end{align}

An important example is when $\bS$ is the Laplacian of a connected undirected graph $G$. In this case, by convention, we apply the canonical polar decomposition $\bS = \bI\bS$, where $\bI$ is the identity matrix. Solving the above optimization results in $\bV_1 = \bV_2$ and $d(\bV_1,\bV_2)=0$. This agrees with what is being considered in GSP. Suppose $\bS'$ is the Laplacian of connected undirected $G'$ that is a small perturbation of $G$. Then we still have a constant $C$ independent of $G'$ such that $\norm{\mathcal{F}_{\bS}-\mathcal{F}_{\bS'}} \leq C\norm{\delta\bS'}_F$ as long as $\bS$ does not have repeated eigenvalues by \cref{thm:tac}.  Similarly, if $\bS$ is the adjacency matrix of the directed cycle graph, then for the canonical polar decomposition $\bS = \bS\bI$, we again have $\bV_1=\bV_2$ by solving (\ref{eq:mbv}) with $d(\bV_1,\bV_2)=0$.  

More generally, the same discussion works if $\bS$ is \emph{normal}, in which case $\bU, \bP$ are simultaneously diagonalizable, and solving (\ref{eq:mbv}) also yields $\bV_1 = \bV_2$ with $d(\bV_1,\bV_2)=0$. 

\subsection{An overall picture}

For simplicity, our discussion in \cref{sec:grp} focuses on the space $\mathcal{S}$ of all invertible and non-derogatory shift operators, which can be weighted. The space $\mathcal{S}$ is endued naturally with the metric $d_{\calS}(\cdot,\cdot)$ induced by the operator norm $\norm{\cdot}$. In this paper, we have defined for each $\bS \in \mathcal{S}$ the Fourier transform operator $\mathcal{F}_{\bS}$. Moreover, for each $k$, we have the convolution operator $\bh_{\bS}^k*$. 

The traditional GSP theory is incomplete by defining GFT (and convolutions) on the subspace of $\mathcal{S}$ consisting of symmetric shift operators. In this paper, we have extended the constructions to all of $\mathcal{S}$, which subsumes GSP as a special case. Moreover, by \cref{thm:tac}, the extension of the traditional theory is in a continuous manner. We call our framework \emph{directed GSP} (DGSP). A comparison with GSP is summarized in \cref{tab:cbt}. 

\begin{table}
\caption{Comparison between traditional GSP and the proposed approach ``DGSP''} \label{tab:cbt}
\begin{center}
\begin{tabular}{ |c|c|c| } 
 \hline
  & GSP & DGSP \\ 
 \hline
 \hline
 Applicability & Symmetric matrices & $\mathcal{S}$ \\ 
 \hline
 Frequency domain  & $\mathbb{R}^n$ & $\mathbb{C}^{n^2}$ \\ 
 \hline
 Fourier transform & $\bV^{\T}\boldf$ & $\bV_1^H\bV_2D(\bV_2^H\boldf)$ \\ 
 \hline
 Inverse transform & $\bV\boldf$ & $\bV_1\bM\bc$ \\ 
 \hline
 Convolution & Polynomial in $\bS$ & Not polynomial in $\bS$ \\
 \hline
 Relation to GSP & $--$  & GSP if $\bS$ is normal \\
 \hline
\end{tabular}
\end{center}
\end{table}

\subsection{General matrix decomposition} \label{sec:gmd}
The Fourier transform (\ref{eq:mfb}) has the element of the chain rule in calculus. More precisely, for a signal $\boldf$, the polar decomposition trivially implies that $\bS\boldf = \bU(\bP\boldf)$. In the expression $\langle \bv_{2,i},\boldf \rangle \langle \bv_{1,j},\bv_{2,i} \rangle$, we view $\langle \bv_{2,i},\boldf \rangle$ as the ``partial derivative of $\boldf$ w.r.t.\ $\bv_{2,i}$, and the factor $\langle \bv_{1,j},\bv_{2,i} \rangle$ as the ``partial derivative'' of $\bv_{2,i}$ w.r.t.\ $\bv_{1,j}$. This gives us hints on how to generalize the framework of the paper.

Instead of considering the polar decomposition $\bS = \bU\bP$, let 
\begin{align}
\bS = \bM_1\cdots \bM_k, k>1
\end{align}
be any decomposition of $\bS$ into a product of normal matrices $\bM_i, 1\leq i\leq k$. Assume that for each $\bM_i$, there is an eigendecomposition $\bM_i = \bV_i\bLambda_i\bV_1^H$ such that $\bV_i$ is unitary and $\bLambda_i$ is the diagonal matrix of eigenvalues. Let the columns of $\bV_j$ be $\{\bv_{i,j}, 1\leq j\leq n\}$. Then the Fourier transform of a signal $\boldf$ is a map $\mathbb{R}^n \to \mathbb{C}^{n^k}$ defined by
\begin{align*}
    \boldf \mapsto \Big(\prod_{1\leq i\leq k-1} \langle \bv_{i,j_i}, \bv_{i+1,j_{i+1}}\rangle\Big)\langle \bv_{k,j_k}, \boldf\rangle, 1\leq j_1,\ldots,j_k \leq n,
\end{align*}
as a generalization of (\ref{eq:mfb}) in \cref{defn:tft}. 

If the eigendecompositions of $\bM_i$'s are non-unique, as in the case of the polar decomposition, we may jointly estimate $\bV_i, 1\leq i\leq k$ by solving the following optimization taking the form analogous to (\ref{eq:mbv}):
\begin{align*}
    \min_{\bV_1,\ldots,\bV_k} \sum_{1\leq i\leq k-1} d(\bV_i, \bV_{i+1}), \text{ provided } \bV_i \text{ forms an eigenbasis of } \bM_i. 
\end{align*}

We notice that the codomain of the Fourier transform is $n^k$ dimensional. Therefore, it will be intractable if $k$ is large. This means that it is only reasonable to consider only the matrix decomposition of a small number of factors. Examples of such a matrix decomposition include the reverse polar decomposition, the singular value decomposition (SVD) \cite{Tao12}, and the Sinkhorn decomposition involving doubly stochastic matrix \cite{Sin64}.  

\section{Numerical results} \label{sec:nur}
 
In this section, we consider the heat flow dataset of the Intel Berkeley Research lab.\footnote{http://db.csail.mit.edu/labdata/labdata.html} Temperature data are collected from $53$ sensors, denoted by $V$, placed in the lab. We use the setup of \cite{Furutani2019GraphSP} and construct a planar graph $G$ with an edge set $E$ of size $87$. Each edge is given a direction such that heat flows from high temperature to low temperature. A signal $\boldf$ on $G$ consists of temperature readings on the sensors $V$. 

We follow \cite[Section~4.2]{Furutani2019GraphSP} and perform the signal denoising task. For a temperature signal $\boldf$, we add a noise $\bn$ whose components are independent Gaussians with mean $\mu=0$ and standard deviation $\sigma$. The task is to recover $\boldf$, while observing $\boldf+\bn$. In \figref{fig:tss}, we show the spectral plots of a true signal $\boldf$ and a noise $\bn$. We see that the spectral of $\boldf$ is mostly supported for small frequency indices as expected. On the other hand, high magnitude components of the spectral coefficients $\bn$ are spread over the entire frequency domain. This suggests a low-pass filter can be used for denoising. As in \cite{Furutani2019GraphSP}, we consider a convolution filter $f(\bh_{\bS})$ (\cref{sec:con}) with $f(x) = 1/(1+cx)$ for a tunable hyperparameter $c$.

\begin{figure}
    \centering
    \includegraphics[scale=0.4, trim={1cm 0.5cm 0.8cm 0},clip]{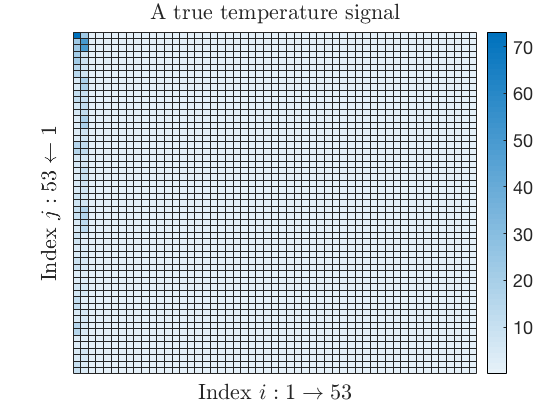}
    \includegraphics[scale=0.4,trim={1cm 0.5cm 0.8cm 0},clip]{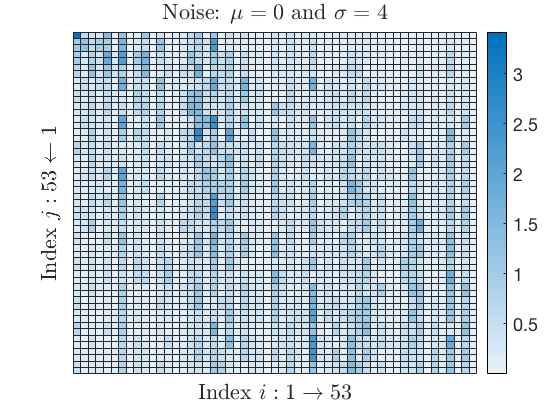}
    \caption{The sample spectral plots of a temperature signal in the lab and a noise signal with $\sigma=4$. The magnitudes of the spectral coefficients are shown.}
    \label{fig:tss}
\end{figure}

Let $\boldf'$ be the recovered signal. We evaluate the performance by computing the \emph{root-mean-squared-error} (RMSE) $\norm{Re(\boldf')-\boldf}_2/\sqrt{53}$ over all the sensors. Here $Re(\boldf')$ is the componentwise real part of $\boldf'$. We call our approach ``DGSP''. We compare with the method ``Hermitian'' of \cite{Furutani2019GraphSP}. The boxplots of RMSEs over $1000$ runs and $\sigma = 1,\ldots, 8$ are shown in \figref{fig:tro}. For reference, we also show the RMSE between $\boldf+\bn$ and $\boldf$, called ``Base''. We see that our lossless directed graph signal processing approach ``DGSP'' performs better for all $\sigma$. 

\begin{figure}
    \centering
    \includegraphics[scale=0.3]{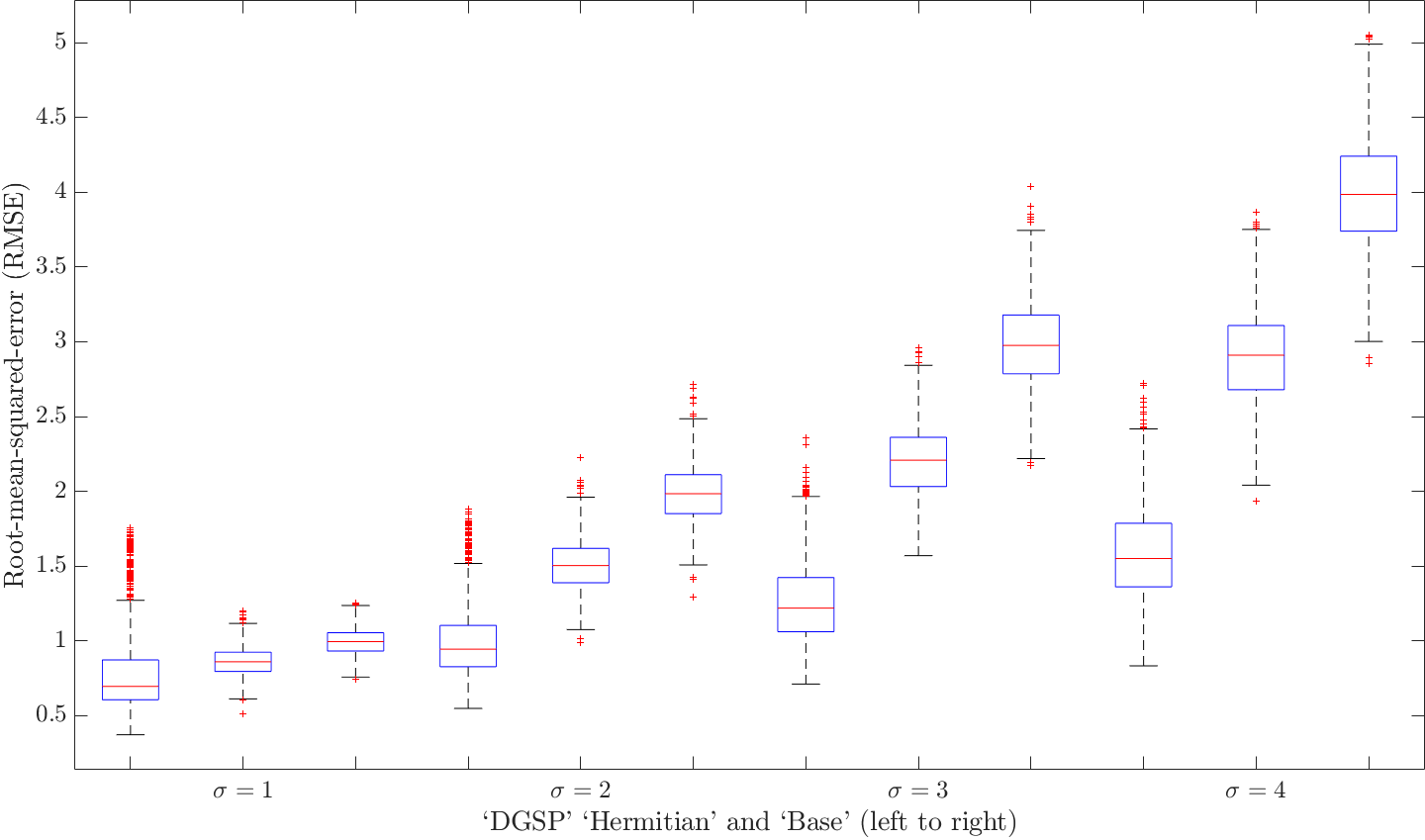}
    \includegraphics[scale=0.3]{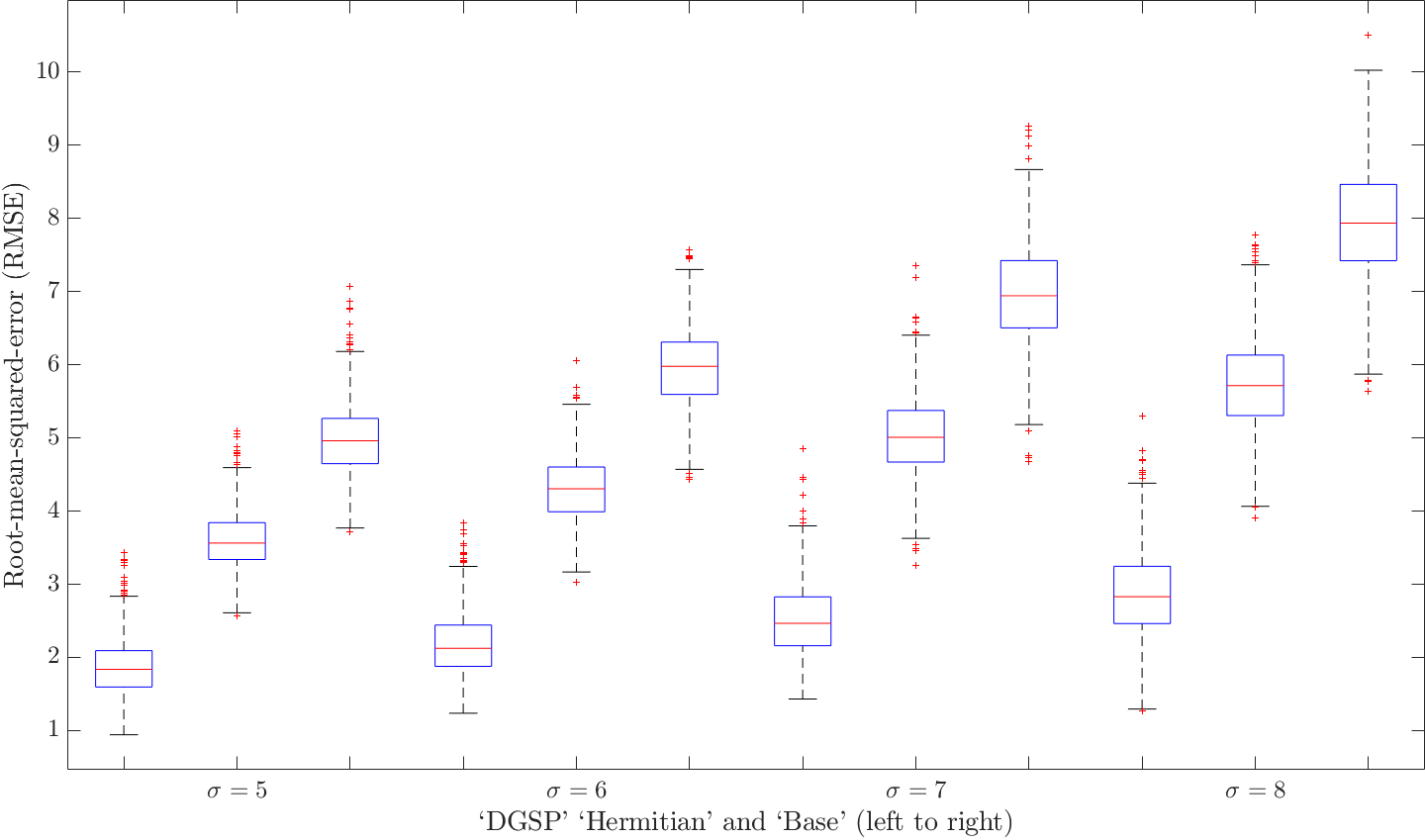}
    \caption{The RMSE of noisy signal recovery.}
    \label{fig:tro}
\end{figure}

\section{Conclusions} \label{sec:conc}
This paper introduces a novel signal processing framework for directed graphs. Our approach utilizes the polar decomposition to introduce a novel Fourier transform with a larger frequency domain than that of traditional GSP. This enables us to define convolution using a standard procedure. Our method possesses two notable features: it ensures losslessness, as the shift operator can be entirely reconstructed from the polar decomposition factors, and it encompasses traditional graph signal processing when applied to directed graphs. For future work, we shall study how the framework might be used for learning directed graph topology. 

\bibliographystyle{IEEEtran}
\bibliography{IEEEabrv,StringDefinitions,refs}

\end{document}

%% file: preamble.tex
\usepackage[T1]{fontenc}
\usepackage{amsmath,amssymb,amsfonts,mathrsfs,bm}
\usepackage{mathtools}
\usepackage{amsthm}
\usepackage{nicefrac}
\usepackage{cite}
\usepackage[shortlabels]{enumitem}
\usepackage{graphicx}
\usepackage{epstopdf}
\usepackage{url}
\usepackage{colortbl}
\usepackage{booktabs}
\usepackage{multirow}
\usepackage[table,dvipsnames]{xcolor}
\usepackage[normalem]{ulem}

\usepackage{array}
\newcolumntype{L}[1]{>{\raggedright\let\newline\\\arraybackslash\hspace{0pt}}m{#1}}
\newcolumntype{C}[1]{>{\centering\let\newline\\\arraybackslash\hspace{0pt}}m{#1}}
\newcolumntype{R}[1]{>{\raggedleft\let\newline\\\arraybackslash\hspace{0pt}}m{#1}}

\makeatletter
\let\MYcaption\@makecaption
\makeatother
\usepackage[font=footnotesize]{subcaption}
\makeatletter
\let\@makecaption\MYcaption
\makeatother

\usepackage{xparse}



\usepackage{glossaries}

\makeatletter
\let\oldgls\gls
\let\oldglspl\glspl

\newcommand\fussy@ifnextchar[3]{%
  \let\reserved@d=#1%
  \def\reserved@a{#2}%
  \def\reserved@b{#3}%
  \futurelet\@let@token\fussy@ifnch}
\def\fussy@ifnch{%
  \ifx\@let@token\reserved@d
    \let\reserved@c\reserved@a 
  \else
    \let\reserved@c\reserved@b
  \fi
 \reserved@c}

\renewcommand{\gls}[1]{%
  \oldgls{#1}\fussy@ifnextchar.{\@checkperiod}{\@}}
\renewcommand{\glspl}[1]{%
  \oldglspl{#1}\fussy@ifnextchar.{\@checkperiod}{\@}}

\newcommand{\@checkperiod}[1]{%
  \ifnum\sfcode`\.=\spacefactor\else#1\fi
}
\makeatother

\newacronym{wrt}{w.r.t.}{with respect to}
\newacronym{RHS}{RHS}{right-hand side}
\newacronym{LHS}{LHS}{left-hand side}
\newacronym{iid}{i.i.d.}{independent and identically distributed}

\usepackage{float}

\ifx\notloadhyperref\undefined
	\ifx\loadbibentry\undefined
		\usepackage[hidelinks,hypertexnames=false]{hyperref} 
	\else
		\usepackage{bibentry}
		\makeatletter\let\saved@bibitem\@bibitem\makeatother
		\usepackage[hidelinks,hypertexnames=false]{hyperref}
		\makeatletter\let\@bibitem\saved@bibitem\makeatother
	\fi
\else
	\ifx\loadbibentry\undefined
		\relax
	\else
		\usepackage{bibentry}
	\fi
\fi

\usepackage[capitalize]{cleveref}
\crefname{equation}{}{}
\Crefname{equation}{}{}
\crefname{claim}{claim}{claims}
\crefname{step}{step}{steps}
\crefname{line}{line}{lines}
\crefname{condition}{condition}{conditions}
\crefname{dmath}{}{}
\crefname{dseries}{}{}
\crefname{dgroup}{}{}

\crefname{Problem}{Problem}{Problems}
\crefformat{Problem}{Problem~(#2#1#3)}
\crefrangeformat{Problem}{Problems~(#3#1#4) to~(#5#2#6)}

\crefname{Theorem}{Theorem}{Theorems}
\crefname{Corollary}{Corollary}{Corollaries}
\crefname{Proposition}{Proposition}{Propositions}
\crefname{Lemma}{Lemma}{Lemmas}
\crefname{Definition}{Definition}{Definitions}
\crefname{Example}{Example}{Examples}
\crefname{Assumption}{Assumption}{Assumptions}
\crefname{Remark}{Remark}{Remarks}
\crefname{Rem}{Remark}{Remarks}
\crefname{remarks}{Remarks}{Remarks}
\crefname{Appendix}{Appendix}{Appendices}
\crefname{Exercise}{Exercise}{Exercises}
\crefname{Theorem_A}{Theorem}{Theorems}
\crefname{Corollary_A}{Corollary}{Corollaries}
\crefname{Proposition_A}{Proposition}{Propositions}
\crefname{Lemma_A}{Lemma}{Lemmas}
\crefname{Definition_A}{Definition}{Definitions}

\usepackage{crossreftools}
\ifx\notloadhyperref\undefined
\else
	\relax
\fi

\usepackage{algorithm,algorithmic}

\ifx\loadbreqn\undefined
	\relax
\else
	\usepackage{breqn} 
\fi

\ifx\renewtheorem\undefined
\ifx\useTheoremCounter\undefined
\newtheorem{Theorem}{Theorem}
\newtheorem{Corollary}{Corollary}
\newtheorem{Proposition}{Proposition}
\newtheorem{Lemma}{Lemma}
\else
\newtheorem{Theorem}{Theorem}

\fi

\newtheorem{Definition}{Definition}
\newtheorem{Example}{Example}

\fi

\theoremstyle{remark}

\theoremstyle{plain}

\newcommand{\calS}{\mathcal{S}}

\newcommand{\calU}{\mathcal{U}}

\newcommand{\bA}{\mathbf{A}}

\newcommand{\bc}{\mathbf{c}}

\newcommand{\bD}{\mathbf{D}}

\newcommand{\boldf}{\mathbf{f}}

\newcommand{\bg}{\mathbf{g}}

\newcommand{\bh}{\mathbf{h}}

\newcommand{\bi}{\mathbf{i}}
\newcommand{\bI}{\mathbf{I}}

\newcommand{\bL}{\mathbf{L}}

\newcommand{\bM}{\mathbf{M}}
\newcommand{\bn}{\mathbf{n}}
\newcommand{\bN}{\mathbf{N}}

\newcommand{\bP}{\mathbf{P}}

\newcommand{\bS}{\mathbf{S}}

\newcommand{\bU}{\mathbf{U}}
\newcommand{\bv}{\mathbf{v}}
\newcommand{\bV}{\mathbf{V}}

\newcommand{\bW}{\mathbf{W}}
\newcommand{\bx}{\mathbf{x}}

\DeclareSymbolFont{bsfletters}{OT1}{cmss}{bx}{n}
\DeclareSymbolFont{ssfletters}{OT1}{cmss}{m}{n}
\DeclareMathSymbol{\bsfGamma}{0}{bsfletters}{'000}
\DeclareMathSymbol{\ssfGamma}{0}{ssfletters}{'000}
\DeclareMathSymbol{\bsfDelta}{0}{bsfletters}{'001}
\DeclareMathSymbol{\ssfDelta}{0}{ssfletters}{'001}
\DeclareMathSymbol{\bsfTheta}{0}{bsfletters}{'002}
\DeclareMathSymbol{\ssfTheta}{0}{ssfletters}{'002}
\DeclareMathSymbol{\bsfLambda}{0}{bsfletters}{'003}
\DeclareMathSymbol{\ssfLambda}{0}{ssfletters}{'003}
\DeclareMathSymbol{\bsfXi}{0}{bsfletters}{'004}
\DeclareMathSymbol{\ssfXi}{0}{ssfletters}{'004}
\DeclareMathSymbol{\bsfPi}{0}{bsfletters}{'005}
\DeclareMathSymbol{\ssfPi}{0}{ssfletters}{'005}
\DeclareMathSymbol{\bsfSigma}{0}{bsfletters}{'006}
\DeclareMathSymbol{\ssfSigma}{0}{ssfletters}{'006}
\DeclareMathSymbol{\bsfUpsilon}{0}{bsfletters}{'007}
\DeclareMathSymbol{\ssfUpsilon}{0}{ssfletters}{'007}
\DeclareMathSymbol{\bsfPhi}{0}{bsfletters}{'010}
\DeclareMathSymbol{\ssfPhi}{0}{ssfletters}{'010}
\DeclareMathSymbol{\bsfPsi}{0}{bsfletters}{'011}
\DeclareMathSymbol{\ssfPsi}{0}{ssfletters}{'011}
\DeclareMathSymbol{\bsfOmega}{0}{bsfletters}{'012}
\DeclareMathSymbol{\ssfOmega}{0}{ssfletters}{'012}

\newcommand{\bLambda}{\bm{\Lambda}}
\newcommand{\bSigma	}{\bm{\Sigma}}

\newcommand{\qednew}{\nobreak \ifvmode \relax \else
      \ifdim\lastskip<1.5em \hskip-\lastskip
      \hskip1.5em plus0em minus0.5em \fi \nobreak
      \vrule height0.75em width0.5em depth0.25em\fi}

\newcommand{\T}{^{\intercal}}

\newcommand{\norm}[1]{{\left\lVert{#1}\right\rVert}}

\DeclareDocumentCommand\set{s m t| m}{%
  \IfBooleanTF#1%
	{\left\{\, #2\mathrel{} \IfBooleanTF{#3}{\middle|}{:}\mathrel{}  #4\, \right\}}%
  {\{\, #2 \IfBooleanTF{#3}{\mid}{\mathrel{} : \mathrel{}} #4\, \}}%
}

\DeclareDocumentCommand \ifcond {m m} {%
	{#1} %
	\IfValueT{#2}{\, \middle|\, {#2}}%
}

\DeclareDocumentCommand \P {e{_} g >{\SplitArgument{ 1 }{ @| }}d() g } {%
	\mathbb{P}%
	\IfValueTF{#1}{_{#1}}
		{\IfValueT{#2}{_{#2}}}%
	\IfValueT{#3}{\left(\ifcond#3}%
	\IfValueT{#4}{\, \middle|\, {#4}}%
	\IfValueT{#3}{\right)}%
}

\DeclareDocumentCommand \E {e{_} g >{\SplitArgument{ 1 }{ @| }}o g } {%
	\mathbb{E}%
	\IfValueTF{#1}{_{#1}}
		{\IfValueT{#2}{_{#2}}}%
	\IfValueT{#3}{\left[\ifcond#3}%
	\IfValueT{#4}{\, \middle|\, {#4}}%
	\IfValueT{#3}{\right]}%
}

\definecolor{gray90}{gray}{0.9}

\ifx\nohighlights\undefined

	\newcommand{\msout}[1]{\text{\color{green} \sout{\ensuremath{#1}}}}
	\newcommand{\del}[1]{{\color{green}\ifmmode \msout{#1}\else\sout{#1}\fi}}

\else

	\newcommand{\msout}[1]{#1}
	\newcommand{\del}[1]{#1}
\fi

\newcommand{\hhide}[1]{}

\renewcommand{\figurename}{Fig.}
\newcommand{\figref}[1]{\figurename~\ref{#1}}
\graphicspath{{./Figures/}} 
\newcommand{\includeCroppedPdf}[2][]{%
    \IfFileExists{./Figures/#2-crop.pdf}{}{%
        \immediate\write18{pdfcrop ./Figures/#2 ./Figures/#2-crop.pdf}}%
    \includegraphics[#1]{./Figures/#2-crop.pdf}}

\ifx\diagnoselabel\undefined
	\relax
\else
	\makeatletter
	 \def\@testdef #1#2#3{%
		 \def\reserved@a{#3}\expandafter \ifx \csname #1@#2\endcsname
		\reserved@a  \else
	 \typeout{^^Jlabel #2 changed:^^J%
	 \meaning\reserved@a^^J%
	 \expandafter\meaning\csname #1@#2\endcsname^^J}%
	 \@tempswatrue \fi}
	\makeatother
\fi
